\documentclass[conference]{IEEEtran}
\usepackage[T1]{fontenc}
\IEEEoverridecommandlockouts
\def\BibTeX{{\rm B\kern-.05em{\sc i\kern-.025em b}\kern-.08em
    T\kern-.1667em\lower.7ex\hbox{E}\kern-.125emX}}
\usepackage[dvipsnames]{xcolor}
\usepackage{color}

\definecolor{PPQAblue}{HTML}{4878CF} 
\definecolor{PPQAgreen}{HTML}{4878CF}
\definecolor{PPQAlightgreen}{HTML}{94CEB9}
\definecolor{PPQAorange}{HTML}{EB5A36}
\definecolor{PPQApink}{HTML}{DC3977} 
\definecolor{PPQApurple}{HTML}{7C1D6F}

\usepackage{pgfplots}
\pgfplotsset{compat = newest}
\usepackage{adjustbox}
\usepackage[utf8]{inputenc}
\usepackage{physics}
\usepackage{bm} 
\usepackage{bbm} 
\usepackage{braket}
\usepackage{dsfont}
\usepackage{amsmath}
\usepackage{amssymb}
\usepackage{amsthm}
\usepackage{thmtools} 
\usepackage{mathtools}
\usepackage[ruled,linesnumbered]{algorithm2e}
\SetArgSty{textnormal}
\usepackage{algpseudocode}
\usepackage{subfiles}
\usepackage{tikz}
\usepackage{graphicx}   
\usepackage{subcaption} 
\usetikzlibrary{quantikz2}
\usetikzlibrary{shapes.geometric,backgrounds,positioning,shapes.geometric,decorations.markings,decorations.pathreplacing,arrows,knots,hobby,angles,quotes,trees}
\tikzset{meter/.append style={draw, inner sep=5, rectangle, font=\vphantom{A}, minimum width=20, line width=.4,
 path picture={\draw[black] ([shift={(.1,.2)}]path picture bounding box.south west) to[bend left=50] ([shift={(-.1,.2)}]path picture bounding box.south east);\draw[black,-latex] ([shift={(0,.1)}]path picture bounding box.south) -- ([shift={(.3,-.1)}]path picture bounding box.north);}}} 
\makeatletter
\newcommand{\gettikzxy}[3]{%
  \tikz@scan@one@point\pgfutil@firstofone#1\relax
  \edef#2{\the\pgf@x}%
  \edef#3{\the\pgf@y}%
}
\makeatother
\usepackage{graphics}
\usepackage{graphicx}
\usepackage{float}
\usepackage[linesnumbered,ruled]{algorithm2e}
\usepackage[colorlinks]{hyperref}
\hypersetup{
    colorlinks  = true,
    citecolor   = PPQAgreen,
    linkcolor   = PPQAgreen,
    urlcolor    = PPQAblue
}

\declaretheorem[style=plain]{theorem}

\declaretheorem[style=plain,sibling=theorem]{lemma}

        \newcommand*{\N}{\mathbb{N}}

        \newcommand*{\controlled}{C} 
        \newcommand*{\signflip}{\mathcal{S}} 

\usepackage{tcolorbox}  
\usepackage{ragged2e}
\usepackage{enumitem}
\usepackage{booktabs} 
\selectfont

\begin{document}

\title{A gate-time feasibility analysis of quantum breadth-first search for maximum flow problems}

\author{\IEEEauthorblockN{Andreea-Iulia Lefterovici}
\IEEEauthorblockA{\textit{Institut f\"ur Theoretische Physik} \\
\textit{Leibniz Universit\"at Hannover}\\
Hannover, Germany \\
\url{https://orcid.org/0009-0007-4810-5927}}
\and
\IEEEauthorblockN{Lara Lelakowski}
\IEEEauthorblockA{\textit{Institut f\"ur Theoretische Physik} \\
\textit{Leibniz Universit\"at Hannover}\\
Hannover, Germany \\
\url{https://orcid.org/0009-0001-0177-6498}}
\and
\IEEEauthorblockN{Michael Perk}
\IEEEauthorblockA{\textit{Department of Computer Science} \\
\textit{TU Braunschweig}\\
Braunschweig, Germany \\
\url{https://orcid.org/0000-0002-0141-8594}}
}

\maketitle

\begin{abstract}
The maximum flow problem asks to find the largest possible flow from a source to a sink in a capacitated network. 
It arises frequently in scheduling, project selection, and as a core subroutine in broader optimisation tasks. 
Classically, it can be efficiently solved using Dinic's algorithm, which repeatedly performs breadth-first search (BFS) and blocking flow computations on the graph.
As a potential candidate for quantum speedups, these BFS subroutines can be naturally replaced with quantum BFS (qBFS), an instantiation of Grover’s search algorithm. 
In this paper, we evaluate the expected performance of qBFS on standard classical datasets.
These instances are too large to be solved directly on current quantum hardware, so we adopt a hybrid benchmarking approach: (i) we run a classical implementation of Dinic's algorithm and isolate the runtime of its BFS subroutines; (ii) we analytically estimate the number of two-qubit gates required to implement qBFS, where we use the classically logged data.
Our results indicate that, on every tested instance, matching classical BFS would require two-qubit gate operation times below the current physical record by one to two orders of magnitude for most instances, and by a factor of about $1.3$ for the smallest, most favourable ones. Since this comparison uses physical rather than error-corrected gate times, the actual gap is larger still.
\end{abstract}


\section{\label{section:Introduction}Introduction}
In recent years, there has been a lot of optimism that quantum computers could solve real-world optimisation problems, especially when it comes to flow problems such as maximal flow through a network~\cite{Ford1956MaximalFlowThroughtANetwork}, job assignment~\cite{Hopcroft1973MaximumMatchingsBipartiteGaphs}, traffic optimisation~\cite{Ahuja1993NetworkFlows}, and network reliability~\cite{Jongen2004ApplicationsMaxFlowMinCut}.
This optimism has been driven by results coming from complexity analysis.
But asymptotic improvements alone do not guarantee a good practical performance, and this already holds true for classical algorithms such as simplex~\cite{Dantzig1948LinearProgrammingInProblemsForTheNumericalAnalysisOfTheFuture}.

This raises a natural question: why should asymptotic improvements translate into practical advantages for quantum algorithms?
First, we cannot measure the practical performance of quantum algorithms for large data sets simply by executing them on current quantum computers. A reason for this is the current infancy of the quantum hardware~\cite{Preskill2018quantumcomputingin}.
Second, benchmarking studies from Cade et al.~\cite{Cade2023QuantifyingGroverSpeedupsBeyondAsymptoticAnalysis}, Ammann et al.~\cite{Lefterovici2023RealisticRuntimeAnalysisForQuantumSimplexComputation} and Brehm and Weggemans~\cite{Brehm2026AssessingFaultTolerandQuantumAdvantage} have shown that asymptotic complexity can fail to capture runtime when using quantum algorithms, even under benevolent assumptions.
We refer to these studies as hybrid benchmarking\footnote{A collection of techniques that fall under this name can be found in \cite{Lefterovici2026HybridBenchmarking}.}.

Hybrid benchmarking techniques go beyond asymptotic complexity studies by providing explicit analytical expressions for resource estimations: gate, cycle or query counts. 
Within this framework, we consider quantum algorithms that have a corresponding classical counterpart, allowing for a one-to-one comparison. Data are collected by executing the classical algorithm and then fed into the analytical expressions for the quantum algorithm.
In this way, we obtain an estimate of its resource requirements. 
This approach enables a more realistic assessment of quantum algorithms, capturing instance-dependent overheads that asymptotic analysis cannot show.

In this paper, we extend the hybrid benchmarking technique and apply it to the maximum flow problem. 
Specifically, we investigate the practicality of a quantum breadth-first search (qBFS) as a subroutine within a maximum flow algorithm, evaluated on standard classical benchmark datasets that currently cannot be executed on quantum hardware.

Classically, there are a variety of algorithms for solving the maximum flow problem. 
Among the fastest known strongly polynomial algorithms is the one by Orlin~\cite{orlin2013maxflow}, which has an asymptotic running time of $O(\abs{V}\abs{E})$.
Nevertheless, we focus on algorithms that explicitly rely on BFS as a core component.
A natural choice is Dinic's algorithm~\cite{dinic1970}, Dinic's algorithm constructs level graphs using BFS and computes blocking flows efficiently. 
Consequently, it provides a suitable and practically efficient representative, allowing us to isolate and evaluate the implementation cost of qBFS without loss of generality.
Other algorithms in this category include Edmonds--Karp~\cite{edmonds1972}, Malhotra, Kumar, and Maheshwari~\cite{malhotra1978} and Galil~\cite{galil1978flow}. 

We investigate a quantum version of Dinic’s algorithm, as proposed by Ambainis and \v{S}palek~\cite{AmbainisSpalek2005QuantumAlgorithmsForMatchingAndNetworkFlows}, in which the BFS procedure is replaced by qBFS. 
The qBFS uses Grover’s algorithm \cite{Grover1996AFastQuantumMechanicalAlgorithmForDatabaseSearch} to perform the levelling of the graph.
However, since the BFS subroutine is implemented identically across all the classical approaches~\cite{dinic1970,edmonds1972,malhotra1978,galil1978flow}, the specific choice of algorithm does not affect our hybrid benchmarking framework.


\


\section{\label{section:Preliminaries}Preliminaries}
\subsection{Flow networks}
A \textit{network} is a directed graph $G(V,E)$ with vertex set $V$ and edge set $E$ (of sizes $\abs{V}$ and $\abs{E}$), and a function $c$ that assigns a non-negative integer capacity to each edge.
A \textit{flow network} is a network where two of the vertices are labelled as source $s$ and sink $t$.

A \textit{flow} is a function $f$ satisfying two constraints: (i) the flow $f$ of an edge cannot exceed its capacity $c$; (ii) the sum of the incoming flow must equal the sum of the outgoing flow at a vertex $u$, except for the source and the sink vertices.
The value of a flow is the total flow leaving the source. A \textit{maximum flow} is a flow of maximum possible value.

The \textit{residual capacity} of an edge is the difference between its capacity and the flow going through it.
A \textit{residual network} $G_R$ contains all the vertices and edges from before, but where the capacities are updated to the residual capacities.

An \textit{augmenting path} is a path from $s$ to $t$ in $G_R$ along the edges whose residual capacity is positive.

A \textit{layered network} of $G$ is a network whose vertices are ordered into layers, such that the edges always go from layer $\ell$ to layer $\ell + 1$.

A \textit{blocking flow} of a layered residual network is a flow such that every path from $s$ to $t$ is saturated on at least one edge.

\subsection{QSearch}
QSearch is a generalisation of Grover's algorithm \cite{Grover1996AFastQuantumMechanicalAlgorithmForDatabaseSearch, Boyer1998TightBoundsOnQuantumSearching} and is designed to perform an unstructured search for finding one of the $t$ marked elements within a list $L$ of elements. 
The search is performed on a register of $n = \lceil \log_2 \abs{L} \rceil$ qubits that binary-encodes the $\abs{L}$ list elements, by iteratively applying the Grover operator $$Q \coloneqq \text{H}^{\otimes n}  \signflip_{0} \text{H}^{\dagger \otimes n} \signflip_{f},$$ where $\text{H}$ is a Hadamard gate, $\signflip_{0}$ is a phase flip around the all-zero state and $\signflip_{f} \ket{\bm{x}} \coloneqq (-1)^{f(\bm{x})} \ket{\bm{x}}$ is the phase oracle,
until one of these $t$ marked elements is found.
The operator $Q$ consists of two reflections: the oracle reflection $\signflip_{f}$, which marks the wanted elements, and the diffusion operator $\text{H}^{\otimes n}\signflip_{0}\text{H}^{\dagger\otimes n}$, which reflects about the uniform superposition $\text{H}^{\otimes n}\ket{0}^{\otimes n}$.
We summarise the procedure in \autoref{alg:qsearch} and we briefly explain it below.

Since the number of marked elements $t$ is unknown, the optimal number of Grover iterations cannot be fixed in advance, hence QSearch samples it as in~\cite{Boyer1998TightBoundsOnQuantumSearching}, with an allowed failure probability $\varepsilon > 0$.
The inner loop of the algorithm draws a power $j$ uniformly at random from $[0,\lfloor m\rfloor]$ and applies $Q^{j}$, where the ceiling $m$ grows geometrically by the factor $\lambda = 6/5$ up to $\sqrt{\abs{L}}$; the counter $k$ (bounded by $k_{max}$) indexes this growth and $k_*$ is the point at which $\lambda^{k}$ reaches the order $\sqrt{\abs{L}}$ needed to rotate a marked element. If a measurement returns no marked element, then the procedure restarts.
The outer loop drives the overall failure probability below $\varepsilon$. If no marked element is found after all restarts, QSearch reports that none exists.

\begin{algorithm}
\caption{QSearch}\label{alg:qsearch}
\SetKwFunction{QSearch}{QSearch}
\SetKwProg{Function}{Function}{}{}
\Function{\QSearch{$f:L\rightarrow\{0,1\}$, $\varepsilon >0$}}{
    $\lambda \gets 6/5$, $m \gets \lambda$\;
    $s\gets 1$, $k \gets 1$, $s_{max}\gets\lceil\log_3{1/\varepsilon}\rceil$\;
    $k_*\gets\Bigl\lceil \log_\lambda\frac{\abs{L}}{2\sqrt{\abs{L}-1}}\Bigr\rceil$, $k_{max}\gets k_*+4$\;
    \While{$s\le s_{max}$}{
        \While{$k\le k_{max}$}{
            Prepare superposition of all vertices\;
            Choose $j$ uniformly at random in $[0,\lfloor m\rfloor]$\;
            Apply $Q^j$\;
            Measure the register, let $y$ be the outcome\;
            \If{$f(y)=1$}{
                \Return $y$\;
            }
            \Else{
                $k\gets k+1$\;
                $m \gets \min(\lambda m, \sqrt{\abs{L}})$\;
            }
        }
        $s\gets s+1$\;
    }
    \Return No marked element\;
}

\end{algorithm}

\begin{algorithm}
\caption{Quantum breadth-first search}\label{alg:qBFS}
\SetKwFunction{qBFS}{qBFS}
\SetKwFunction{QSearch}{QSearch}
\SetKwProg{Function}{Function}{}{}
\SetKwRepeat{Do}{do}{while}%

\Function{\qBFS}{$G=(V, E)$, $s \in V$}{
    $\ell(s) \gets 0$\;
    $\ell(x) \gets \infty$ for $x \neq s$\;
    Create a one-entry queue $W \gets \{s\}$\;

    \While{$W \neq \emptyset$}{
        $x \gets$ pop($W$)\;
        
        \While{True}{
            $y \gets $ \QSearch$(f_x, \varepsilon)$\;
            \If{$y$ is not found}{
                \textbf{break}\;
            }
            $\ell(y) \gets \ell(x) + 1$\;
            Push $y$ to $W$\;
        }
    }
}
\end{algorithm}

\begin{figure*}[htbp]
    \centering
    \begin{subfigure}[b]{0.45\textwidth}
        \centering
        \includegraphics[width=\textwidth]{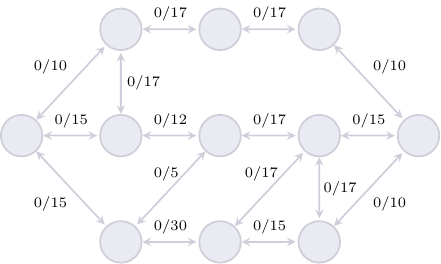} 
        \caption{}
        \label{fig:a}
    \end{subfigure}
    \hfill
    \begin{subfigure}[b]{0.45\textwidth}
        \centering
        \includegraphics[width=\textwidth]{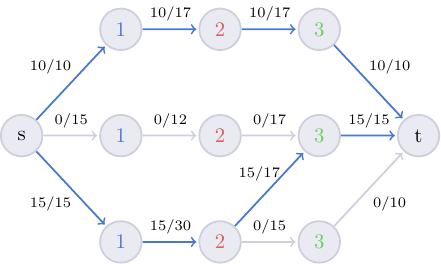}
        \caption{}
        \label{fig:b}
    \end{subfigure}

    \begin{subfigure}[b]{0.45\textwidth}
        \centering
        \includegraphics[width=\textwidth]{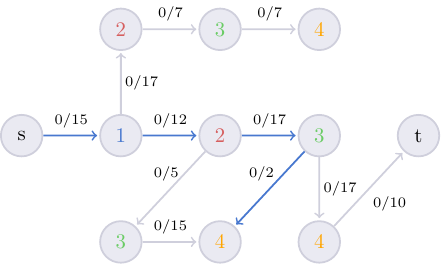}
        \caption{}
        \label{fig:c}
    \end{subfigure}
    \hfill
    \begin{subfigure}[b]{0.45\textwidth}
        \centering
        \includegraphics[width=\textwidth]{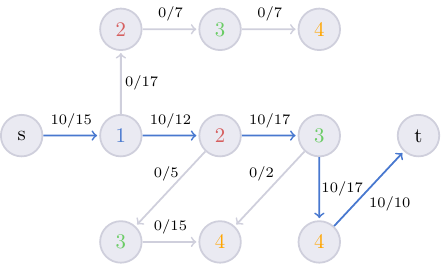}
        \caption{}
        \label{fig:d}
    \end{subfigure}

    \caption{Application of Dinic's algorithm to find the maximum flow for the graph shown in (a). The BFS finds, as shown in (b), three layers, depicted in blue, red and green, respectively.
    The first path (up) is found by performing a DFS and gives a flow value of $f_1 = \min(10, 17, 17, 10) = 10$.
    The second path (down) yields a flow of $f_2 = \min(15, 30, 17, 15) = 15$.
    At this point, no additional paths can be found, implying that a blocking flow has been reached.
    The first block of iteration stops.
    The second block of iteration is depicted in (c) and (d), and takes as an input an updated version of the original graph, where the remaining capacities have been adjusted.
    For simplicity, we removed the edges with remaining capacity equal to zero.
    Dinic's algorithm reiterates, resulting in a different layered graph that has four levels.
    As illustrated in (c), following the blue line leads to a dead-end, as the sink has not been reached.
    In such situations, the algorithm backtracks, while pruning the infeasible edges, and continues until the sink is reached, as shown in (d). 
    The path yields a flow $f_3 = \min(15, 12, 17, 17, 10) = 10$.
    The maximum flow $f = f_1 + f_2 + f_3 = 35$ has been reached, because the edges reaching the sink are saturated.}
    \label{figure:MaxFlowExample}
\end{figure*}

\subsection{Quantum breadth-first search}
We consider the problem of layering the vertices of a directed graph $G = (V, E)$ with a starting vertex $s \in V$.
The goal is to assign layers $\ell : V \rightarrow \N \cup \{\infty\}$ to its vertices such that
$\ell(s) = 0$ and $\ell(y) = 1 + \min_{x:(x,y)\in E}\ell(x)$ otherwise.
Here $\ell$ is defined only on the vertices reachable from $s$; any vertex not reachable from $s$ is assigned $\ell(y) = \infty$.

This layering problem is exactly what BFS solves: starting from $s$, it explores the graph level by level, assigning each vertex the smallest distance from $s$ by visiting all neighbours at distance $\ell$ before moving to distance $\ell+1$.

BFS maintains a queue of discovered vertices and processes them in order, ensuring that each vertex is assigned the correct layer, having a runtime complexity $O(\abs{V}+\abs{E})$.

Ambainis and \v{S}palek \cite{AmbainisSpalek2005QuantumAlgorithmsForMatchingAndNetworkFlows} proposed a quantum version of the BFS that runs in time $O\left(\abs{V}^{3 / 2} \log \abs{V}\right)$ in the adjacency model and in time $O(\sqrt{\abs{V} \abs{E}} \log \abs{V})$ in the list model.
The two query models differ in how the graph is accessed: in the \emph{adjacency model}, a single query returns whether $(u,v)\in E$ for a queried pair of vertices $u,v$, while in the \emph{list model}, a query returns the $i$-th out-neighbour of a queried vertex together with its degree.
The difference between BFS and qBFS lies in how the neighbours of a vertex are discovered.
Classical BFS enumerates all neighbours explicitly, while qBFS repeatedly calls QSearch to find neighbours.
We summarise the method in \autoref{alg:qBFS}, where $f_x : V \rightarrow \{0,1\}$ marks the unvisited neighbours of the current vertex $x$, i.e.\ $f_x(y) = 1$ if and only if $(x,y) \in E$ and $\ell(y) = \infty$.

\subsection{\label{subsection:Assumptions}Assumptions}
Our cost estimate for qBFS relies on the following assumptions, all of which favour the quantum algorithm. They are motivated in \autoref{section:Method} and their effect is shown  in \autoref{section:Results}.
\begin{enumerate}
    \item \emph{We count only two-qubit gates.} We decompose the multi-controlled $Z$ gate of the diffusion operator into two-qubit gates and do not count single-qubit gates.
    \item \emph{Ideal fault-tolerant quantum computer.} We assume full qubit connectivity and no error-correction overhead.
    \item \emph{Per-layer search.} The QSearch calls are aggregated per layer rather than performed per dequeued vertex as in \autoref{alg:qBFS}, which avoids the final unsuccessful run of each per-vertex search and lowers the iteration count.
    \item \emph{Identical layered graph.} qBFS finds every reachable vertex, following the same steps as BFS.
    \item \emph{Unoptimised classical baseline.} Dinic's algorithm runs without optimisations.
\end{enumerate}

\section{Dinic's algorithm and its quantum variant}
Dinic's algorithm is one of the most used algorithms to solve the maximum flow problem, having a strongly polynomial runtime of $O(\abs{V}^2\abs{E})$, but being extremely fast in practice, where it typically runs well below its worst-case bound on realistic networks~\cite{Ahuja1993NetworkFlows, CherkasskyGoldberg1997PushRelabel}.
The algorithm works by guiding augmenting paths from $s$ to $t$ using a layered graph.
We summarise the steps of the algorithm:
\begin{enumerate}
    \item The algorithm uses BFS to construct a layered network.
    \item Then, it uses multiple depth-first search (DFS) procedures to identify augmenting paths from the source $s$ to the sink $t$ until a blocking flow is reached, by always going from layer $\ell$ to layer $\ell + 1$.
\end{enumerate}
In \autoref{figure:MaxFlowExample} we provide an example of a flow network with 11 vertices and 15 edges on which we illustrate the action of Dinic's algorithm that finds a maximum flow of $f = 35$.

The quantum version of Dinic's algorithm replaces the BFS in step one with qBFS, as suggested by Ambainis and \v{S}palek~\cite{AmbainisSpalek2005QuantumAlgorithmsForMatchingAndNetworkFlows}.
They showed that a maximum flow in a network with integer capacities at most $C$ can be found in time $O\left(\abs{V}^{13 / 6} \cdot C^{1 / 3} \log \abs{V}\right)$ in the adjacency model and in time $O\left(\min \left(\abs{V}^{7 / 6} \sqrt{\abs{E}} \cdot C^{1 / 3}, \sqrt{\abs{V} C}\, \abs{E}\right) \log \abs{V}\right)$ in the list model.


\section{\label{section:Method}Methods}

\begin{figure*}
    \centering
    \includegraphics[width=\linewidth]{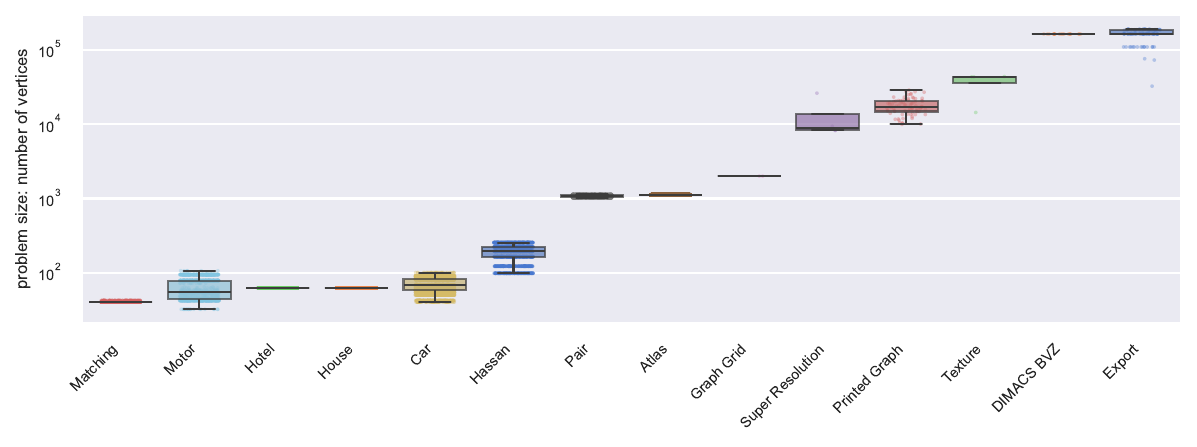}
    \caption{Standard maximum flow benchmark dataset \cite{Jensen2022, Verma_Batra} with sizes between $10^1$ and $10^5$ vertices.}
    \label{fig:problem_instances}
\end{figure*}

\begin{lemma}
\label{lemma:gatecount}
Let $L$ be a list of length $\abs{L}$ with $t$ marked items, binary-encoded on $n = \lceil \log_2 \abs{L} \rceil$ qubits.
A single Grover iteration uses $12(n-1)$ two-qubit gates.
\end{lemma}

\begin{proof}
The Grover iteration $Q = \text{H}^{\otimes n} \signflip_{0} \text{H}^{\dagger\otimes n} \signflip_{f}$ acts on the $n = \lceil \log_2 \abs{L} \rceil$ qubits. We decompose it under very generous assumptions into elementary gates and count the two-qubit gates; every other operation is single-qubit or the black-box oracle.

\begin{enumerate}
    \item Oracle $\signflip_{f}$: assumed to be a single gate; contributes no two-qubit gates.
    \item Diffusion $\text{H}^{\otimes n} \signflip_{0} \text{H}^{\dagger\otimes n}$: the two $\text{H}^{\otimes n}$ layers are single-qubit Hadamards. The reflection $\signflip_{0} = X^{\otimes n}\, \controlled^{\,n-1}Z\, X^{\otimes n}$ about $\ket{0}^{\otimes n}$ has single-qubit $X$ layers, and its only entangling part is the $(n-1)$-controlled $Z$ gate on the $n$ qubits. For $n \ge 3$ this gate is \emph{non-Clifford} and therefore \emph{cannot} be built from CNOT and CZ gates alone. We realise it with the standard textbook construction~\cite{NielsenChuang2010} (Fig.~4.10), an ancilla-assisted AND-tree of $2(n-1)$ Toffoli gates, where $n-1$ are used to compute the conjunction of the $n$ controls into ancillas and $n-1$ to uncompute (together with a single-qubit phase). Each Toffoli decomposes into $6$ CNOTs and single-qubit gates (Fig.~4.9 in~\cite{NielsenChuang2010}), which is CNOT-optimal for an exact Toffoli~\cite{Shende2009CNOTcostToffoli}.
\end{enumerate}

Better Toffoli decomposition techniques exist, e.g.\ relative-phase Toffolis~\cite{Maslov2016RelativePhaseToffoli} or measurement-based uncomputation~\cite{Gidney2018HalvingAddition}, and they reduce the two-qubit count by a small constant factor. 
We do not consider them here because such a factor could bring only the smallest instances, which lie within a factor of about $1.3$ of the record (\autoref{section:Results}), to break even, while the bulk of the dataset remains one to two orders of magnitude below it. 
Moreover, the comparison is made against a physical gate time, whereas instances of this size would require error-corrected logical gates, which are several orders of magnitude slower~\cite{Babbush2021FocusBeyondQuadratic}.

The two-qubit gates thus arise from the AND-tree, giving $6 \cdot 2(n-1) = 12(n-1)$ two-qubit gates per iteration. 
Multiplying by $N_Q(\abs{L}, t)$ we get the total two-qubit-gate count.
\end{proof}

\begin{lemma}
\label{lemma:qsearch-all}
Let $L$ be a list of length $\abs{L}$, with $t$ marked items.
The expected number of iterations that QSearch needs to find \emph{all} $t$ marked items is
\begin{align*}
N_Q(\abs{L}, t) &= \sum_{i=0}^{t-1} n_Q(\abs{L}-i,\, t-i),
\end{align*}
where
\begin{align*}
n_Q(\abs{L}, t)&=\sum_{k=1}^{k_{max}}\frac{m_k}{2}\prod_{l=1}^{k-1}\Bigl(\frac{1}{2}+\frac{\sin(4(m_l+1)\theta)}{4(m_l+1)\sin(2\theta)}\Bigr), \\
k_{max}&=\Bigl\lceil \log_\lambda\frac{\abs{L}}{2\sqrt{\abs{L}-1}}\Bigr\rceil+4\,,
\end{align*}
with $\sin^2\theta=t/\abs{L}$, $m_k = \lfloor \min(\lambda^k, \sqrt{\abs{L}})\rfloor$, $\lambda=6/5$.
\end{lemma}

\begin{proof}
The expression for $n_Q(\abs{L}, t)$ follows from~\cite{Cade2023QuantifyingGroverSpeedupsBeyondAsymptoticAnalysis, Lefterovici2026HybridBenchmarking}.

The inner loop grows the sampling ceiling geometrically as $m_k = \lfloor \min(\lambda^k, \sqrt{\abs{L}}) \rfloor$ with $\lambda = 6/5$, and $k_* = \lceil \log_\lambda \tfrac{\abs{L}}{2\sqrt{\abs{L}-1}} \rceil$ is the point at which $\lambda^{k}$ first reaches the order $\sqrt{\abs{L}}$ required to rotate a single marked item into the marked subspace with constant probability. The additive cutoff constant in $k_{max} = k_* + 4$ ensures that the overall success probability is bounded below by a constant \cite{Cade2023QuantifyingGroverSpeedupsBeyondAsymptoticAnalysis, Boyer1998TightBoundsOnQuantumSearching}.

To find all $t$ marked items, QSearch is applied repeatedly.
After each successful identification of a marked element, that element is removed from the list, reducing both the list size and the number of marked items by one. 
The search register still spans all $\abs{L}$ items and only the oracle changes; evaluating $n_Q$ at $\abs{L}-i$ instead of $\abs{L}$ slightly lowers the cost and is therefore an assumption favourable to qBFS.
Hence, the expected number of iterations is given by summing the cost of finding one marked item over the successive instances $(\abs{L}-i, t-i)$, yielding
\begin{align*}
N_Q(\abs{L}, t) = \sum_{i=0}^{t-1} n_Q(\abs{L}-i,\, t-i).
\end{align*}
\end{proof}

Note that $n_Q$ is the expected number of iterations for having a single successful Grover search when the number of marked elements $t$ is unknown.
For benchmarking purposes, however, we must assume that $t$ is known and computed classically via BFS; otherwise, the hybrid benchmarking method cannot be used. 
In reality, $t$ is unknown, which makes the analysis more complicated and the actual implementation of qBFS more subtle.

In the context of Dinic's algorithm, we need to find \emph{all} marked elements in each layer. 
At each BFS phase, the algorithm constructs a layered graph by assigning every vertex a level corresponding to its shortest-path distance from the source in the residual graph. 
The subsequent DFS phase pushes flow along paths in the layered graph. 
Missing even a single vertex would lead to an incomplete layered graph, in which augmenting paths through that vertex are missed.

Therefore, for each layer $\ell$ containing $t$ vertices out of the $\abs{L}$ vertices of the layered graph, QSearch must be applied until all $t$ marked vertices are found. 
Note that \autoref{alg:qBFS} calls QSearch once per dequeued vertex $x$, searching for the unvisited neighbours of $x$. In our cost model we instead aggregate these calls per layer: a single QSearch over all $\abs{L}$ vertices, whose oracle marks the unvisited vertices with an edge from the current layer, is repeated until all $t$ vertices of the next layer are found. Since the cost of finding $t$ items grows sublinearly in $t$, and each per-vertex search additionally ends with an unsuccessful run, the per-layer count is smaller than the per-vertex one and is therefore favourable to qBFS.
Throughout, we work in the adjacency model defined in \autoref{section:Preliminaries}: the search space is the vertex set of the current layered graph, and the oracle marks a vertex $y$ if it is unvisited and adjacent to a vertex of the current layer. As in \autoref{lemma:gatecount}, this oracle is counted as a single gate.
The expected cost is $N_Q(\abs{L}, t)$ from \autoref{lemma:qsearch-all}. 
Summing over all layers and BFS applications yields the total expected number of quantum iterations required for the BFS component of Dinic's algorithm.

Combining \autoref{lemma:gatecount} and \autoref{lemma:qsearch-all}, the total expected two-qubit-gate count for finding all marked items is
\begin{equation}\label{equation:gatecount}
\begin{aligned}
G_Q(\abs{L}, t) &= c_Q(\abs{L}) \cdot N_Q(\abs{L}, t), \\
c_Q(\abs{L}) &= 12(n-1), \quad n = \lceil \log_2 \abs{L} \rceil.
\end{aligned}
\end{equation}
where $c_Q(\abs{L})$ is the per-iteration two-qubit-gate count of \autoref{lemma:gatecount}.

In the hybrid benchmarking approach, we classically collect the data needed to evaluate this gate count.
In our case, we use a standard implementation of Dinic's algorithm. 
Specifically, during each BFS phase of Dinic's algorithm, we record the number of vertices in each layer, which corresponds to $t$, the number of marked items that QSearch would need to find in that layer.
We further record the number of vertices in the layered graph of each BFS phase, which corresponds to $\abs{L}$.

We assume that each application of qBFS would construct the same layered graph as the classical BFS.
Under this assumption, the classical layer sizes provide the exact inputs needed for \eqref{equation:gatecount} to compute the expected minimum gate count, without actually running the quantum algorithm.
This assumption is justified by the fact that the BFS layering $\ell(v) = \operatorname{dist}_s(v)$ is the shortest-path distance from $s$ in the residual graph. This is a property of the graph, and is independent of the discovery order or of the search method used to explore it. Hence, conditioned on qBFS discovering every reachable vertex, its layered graph is \emph{identical} to the classical one.

We count only an unavoidable subset of the operations required by qBFS; consequently, within the stated encoding and search construction, the resulting two-qubit-gate count is an optimistic gate count of a complete implementation based on this construction.

\begin{figure*}
    \centering
    \includegraphics[width=\linewidth]{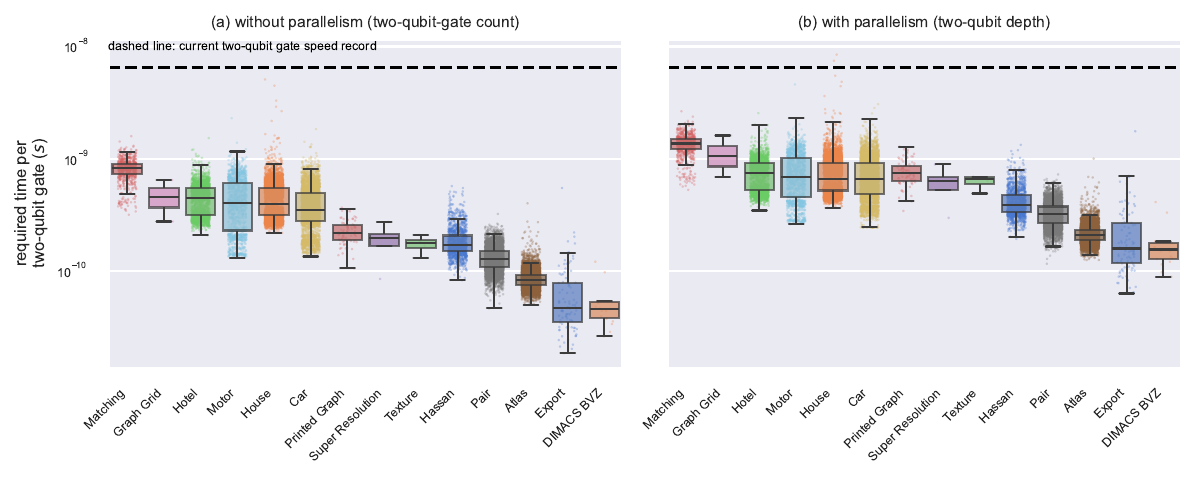}
    \caption{Minimum required time per two-qubit gate for qBFS to match the execution time of classical BFS within Dinic's algorithm, by benchmark-instance type. Panel (a) counts the two-qubit gates with no parallelism assumed, $c_Q(\abs{L}) = 12(n-1)$; panel (b) assumes full parallelism and uses the two-qubit depth $c_Q(\abs{L}) = 12\lceil\log_2 n\rceil$, with $n = \lceil\log_2\abs{L}\rceil$ (\autoref{lemma:gatecount}). The dashed line at $6.5\times10^{-9}$s is the current best achieved two-qubit gate time~\cite{chew2022ultrafast}. In (a) no instance reaches the record; under the optimistic parallel model (b) only two small instances do.}
    \label{fig:results_maxflow}
\end{figure*}


\section{\label{section:Results}Results}
\subsection{Benchmark instances}

We use the standard maximum flow benchmark dataset \cite{Jensen2022, Verma_Batra}, which is a standard test suite for comparing classical algorithms.
The instances are shown in \autoref{fig:problem_instances}.
The dataset contains more than 30,000 problems with a size ranging from $10^1$ to $10^5$ vertices. 

All experiments were run on an Apple M3 processor with 16 GB RAM running macOS Tahoe 26.3.1.
The datafiles and the source code are available here\footnote{\url{https://gitlab.ibr.cs.tu-bs.de/mperk/max-flow/}}.
The classical implementation of Dinic's algorithm is standard, i.e., it uses no optimisation techniques.
The same BFS runtimes can be achieved or exceeded on any modern workstation.

We log the number of vertices $\abs{L}$ in the layered graph of each BFS phase, and the runtime for each BFS step together with the number of vertices $t$ in each layer. We then combine \autoref{lemma:gatecount} and \autoref{lemma:qsearch-all} as in \eqref{equation:gatecount} to estimate the minimum number of gates required to find all marked elements, as explained in \autoref{section:Method}. 

\subsection{Required gate time}

In \autoref{fig:results_maxflow} we show the minimum two-qubit gate time required for qBFS to match the classical BFS runtime on each instance. For each instance, the required time is the ratio between the runtime of its BFS steps and the corresponding two-qubit-gate count $G_Q$ of these steps.
Across the entire dataset it ranges from approximately $1.9 \times 10^{-11}$s to $5.1 \times 10^{-9}$s.
The current physical record of $6.5 \times 10^{-9}$s for a two-qubit gate~\cite{chew2022ultrafast} is slower than required for every instance in the dataset\footnote{When optimally parallelising the gates, two instances could potentially be solved faster on a quantum computer. Nevertheless, such optimism will likely vanish when adding hardware constraints.}, so qBFS does not achieve a runtime advantage on any tested instance.
The gap is nonetheless modest: the most favourable (small) instances come within a factor of about $1.3$ of the record, while the bulk lie one to two orders of magnitude below it.

\subsection{Sensitivity to the cost-model assumptions}
Our estimate makes assumptions that are favourable to the quantum algorithm: the oracle is implemented using a single gate, full qubit connectivity, and found vertices are removed from the search space (\autoref{lemma:qsearch-all}). 
\autoref{tab:sensitivity} shows the gate cost under different ideal assumption models. Considering each Toffoli as one native gate and assuming full parallelism (a circuit-depth model) leads to $26$ instances where the quantum algorithm could be competitive with the classical one; decomposing the Toffolis into two-qubit gates but still allowing full parallelism (two-qubit depth) leaves only $2$; and the precise two-qubit-gate count of \autoref{lemma:gatecount} leaves none. 
Adding a non-trivial oracle or connectivity constraints only reinforces the last verdict. The negative conclusion is hence robust: an advantage appears conditionally and only under the most favourable modelling, and vanishes as soon as the Toffolis are decomposed using two-qubit gates.

Moreover, every physical constraint abstracted away here can only \emph{increase} the two-qubit-gate count $G_Q$, and therefore only lowers the required gate time further below the record. The results reported here are consequently a best case for qBFS.

\begin{table}[t]
\centering
\caption{Sensitivity of different cost models. $c_Q$ is the two-qubit-gate cost per Grover iteration ($n = \lceil\log_2\abs{L}\rceil$); the last column counts, out of more than $30{,}000$ instances, those whose required two-qubit gate time reaches the $6.5\times10^{-9}$s record. The last row is the metric of \autoref{lemma:gatecount}.}\label{tab:sensitivity}
\begin{tabular}{lcc}
\toprule
Cost model & $c_Q$ & $\ge$ record \\
\midrule
Toffoli as native gate, parallel (depth) & $2\lceil\log_2 n\rceil + 6$ & $26$ \\
Toffoli $\to 6$ CNOT, parallel (2q depth) & $12\lceil\log_2 n\rceil$ & $2$ \\
Toffoli $\to 6$ CNOT, 2q count & $12(n-1)$ & $0$ \\
\bottomrule
\end{tabular}
\end{table}


\section{\label{section:DiscussionAndConclusion}Discussion and Conclusions}
In this work, we extend the hybrid benchmarking technique to evaluate the performance of quantum breadth-first search (qBFS) when used as a subroutine in Dinic's maximum flow algorithm. 
We use qBFS as a case study to assess the feasibility of replacing classical breadth-first search (BFS) in Dinic's maximum flow algorithm on large-scale, standard benchmark datasets. 
In each phase, Dinic's algorithm reaches a blocking flow, while the maximum flow is obtained only after the sequence of phases ends.
The level graph of each phase is fixed by the shortest-path distances from the source in the residual graph and does not depend on how the vertices are discovered. Hence, we assume that qBFS produces the same layer structure as BFS.
We analytically derived formulas to estimate the number of gates required by the quantum procedure.  
By inputting instance-specific data from classical runs of Dinic's algorithm into these equations, we obtained gate count estimates for qBFS and their respective execution times.

We observed that qBFS does not outperform classical BFS on any tested maximum flow instance. 
We counted the two-qubit gates in the Grover iterations, an optimistic count on the total number of gates needed to implement the algorithm, and observed that the required two-qubit gate time stays below the gate speed record of $6.5\times10^{-9}\,\mathrm{s}$~\cite{chew2022ultrafast} on the entire dataset. Our count is deliberately optimistic (it omits connectivity, routing, a realistic oracle, and error correction) and even so no tested instance admits a runtime advantage.
Since the results are negative even under assumptions favourable to the quantum algorithm, we do not consider additional overheads such as error correction, which would only widen the gap. Instances of this size could only run on error-corrected hardware, where a logical non-Clifford gate takes microseconds or longer~\cite{Babbush2021FocusBeyondQuadratic}, i.e.\ more than two orders of magnitude longer than the physical record used here. Counting only the $2(n-1)$ logical Toffolis per iteration, which dominate there, the required time per logical gate is still $30$ to $10^4$ times shorter than a microsecond.

Two aspects deserve emphasis. First, our classical baseline is an unoptimised Dinic implementation; a faster baseline would only lower the measured BFS runtime and hence tighten the gate-time requirement on qBFS, so this choice is conservative in favour of the quantum algorithm. Second, qBFS provides only a quadratic (Grover-type) speedup, and such speedups are precisely the ones most easily erased by fault-tolerance overheads.
This result reinforces a broader lesson: asymptotic quantum speedups do not automatically translate into practical advantages.

\section*{Acknowledgments}
This work was supported by the DFG under Germany’s Excellence Strategy - EXC-2123 QuantumFrontiers-2 - 390837967 and SFB 1227 (DQ-mat), the Quantum Valley Lower Saxony, and the BMBF projects Quics, QuBRA and ProvideQ. AIL's research was further supported by the Research Council of Lithuania under the Program ‘University Excellence Initiatives’ of the Ministry of Education, Science and Sports of the Republic of Lithuania (Measure No. 12-001-01-01-01 ‘Improving the Research and Study Environment’), Project No. S-A-UEI-23-11. Helpful correspondence and discussions with Lennart Binkowski, Gabriela Ciolacu and Tobias J. Osborne are gratefully acknowledged.

\bibliographystyle{unsrt}
\bibliography{main.bib}


\end{document}